\documentclass[11pt]{article}

\usepackage{float}
\usepackage{amsmath}
\usepackage{graphicx}
\usepackage[round]{natbib}
\usepackage{float}
\usepackage{geometry}
\usepackage{tikz}
\usepackage[english]{babel}
\usepackage{longtable}
\usepackage{color}
\usepackage{amssymb, amsmath, amsthm, mathrsfs}
\usepackage{multirow}
\usepackage[titletoc,toc,title]{appendix}
\usepackage{authblk}
\usepackage{setspace}
\usepackage{dsfont}
\usepackage{enumerate}
\RequirePackage[OT1]{fontenc}
\usepackage{subcaption}

\DeclareMathOperator{\sll}{sl}
\DeclareMathOperator{\orr}{or}
\DeclareMathOperator*{\argmax}{arg\,max}

\DeclareMathOperator*{\argmin}{arg\,min}
\DeclareMathOperator{\sign}{sign}
\RequirePackage[colorlinks,citecolor=blue,urlcolor=blue]{hyperref}
\newtheorem{remark}{Remark} \newtheorem{theorem}{Theorem} {
  \theoremstyle{plain} \newtheorem{assumption}{} } {
  \theoremstyle{plain}  }

\newtheorem{lemma}{Lemma} 
\newtheorem{definition}{Definition}

\DeclareMathOperator{\err}{\mathcal{E}}
 \DeclareMathOperator{\var}{Var}
 
\newcommand{\E}{\mathbb{E}}

\newcommand{\indep}{\mbox{$\perp\!\!\!\perp$}} 
 
\newcommand{\Pn}{P_n}

\newcommand{\thetaSL}{\hat\theta_{\sll}}
\newcommand{\thetaOR}{\hat\theta_{\orr}}
\newcommand{\dSL}{\hat d_{\sll}}
\newcommand{\dOR}{\hat d_{\orr}}
\newcommand{\fSL}{\hat f_{\sll}}
\newcommand{\fOR}{\hat f_{\orr}}

\newcommand{\thetac}{\theta_{\mbox{\footnotesize c}}}
\newcommand{\betac}{V_{\mbox{\footnotesize c}}}

\newcommand{\Sc}{S_{\mbox{\footnotesize c}}}
\newcommand{\dopt}{d_0}
\newcommand{\doptz}{d_0}
\newcommand{\doptc}{d_{\mbox{c}}}
\newcommand{\thetacm}{\theta_{\mbox{\footnotesize c,m}}}
\newcommand{\one}{\mathds{1}} \newcommand{\R}{\mathbb{R}}
\usepackage{graphicx}
\usepackage{adjustbox}
\usepackage{array}

\newcolumntype{R}[2]{%
  >{\adjustbox{angle=#1,lap=\width-(#2)}\bgroup}%
  l%
  <{\egroup}%
}
\newcommand*\rot{\multicolumn{1}{R{90}{-1em}}}

\renewenvironment{proof}{{\it Proof }}{\qed \\} \title{Targeted
  Learning Ensembles for Optimal Individualized Treatment Rules with
  Time-to-Event Outcomes} \author[1]{Iv\'an D\'iaz\thanks{corresponding
    author: ild2005@med.cornell.edu}} \affil[1]{\small Division of
  Biostatistics, Weill Cornell Medicine.}  \author[1]{Oleksandr
  Savenkov}
\author[1]{Karla Ballman}

\begin{document}\maketitle

\begin{abstract}
  We consider estimation of an optimal individualized treatment rule
  from observational and randomized studies when a high-dimensional
  vector of baseline variables is available. Our optimality criterion
  is with respect to delaying expected time to occurrence of an event
  of interest (e.g., death or relapse of cancer). We leverage
  semiparametric efficiency theory to construct estimators with
  desirable properties such as double robustness. We propose two
  estimators of the optimal rule, which arise from considering two
  loss functions aimed at (i) directly estimating the conditional
  treatment effect (also know as the blip function), and (ii)
  recasting the problem as a weighted classification problem that uses
  the 0-1 loss function. Our estimated rules are \textit{super
    learning} ensembles that minimize the cross-validated risk of a
  linear combination in a user-supplied library of candidate
  estimators. We prove oracle inequalities bounding the finite sample
  excess risk of the estimator. The bounds depend on the excess risk
  of the oracle selector and a doubly robust term related to
  estimation of the nuisance parameters. We discuss some important
  implications of these oracle inequalities such as the convergence
  rates of the value of our estimator to that of the oracle
  selector. We illustrate our methods in the analysis of a phase III
  randomized study testing the efficacy of a new therapy for the
  treatment of breast cancer.
\end{abstract}


\section{Introduction}

Individualized treatment rules play a fundamental role in the
precision medicine model for healthcare, whereby medical decisions are
targeted to the individual based on their expected clinical response,
instead of the traditional one-size-fits-all approach. Mathematically,
a treatment rule is a function that maps an individual's pre-treatment
covariates into an optimal treatment choice. In this paper, we are
concerned with learning the optimal rules from data collected as part
of an observational or randomized study, where optimality is defined
as the maximum delay in the expected time of occurrence of an
undesirable event (e.g., death or relapse).

Recent advances in biomedical imaging and gene expression technology
produce large amounts of data that can be used to tailor treatment to
very specific patient characteristics.  Methods to estimate the
optimal rule when it is defined with respect to a single time-point
outcome include the work of \cite{qian2011performance,
  zhao2012estimating, song2015sparse, rubin2012statistical,
  mckeague2014estimation}, among others. Methods to solve the problem
using survival outcomes subject to informative censoring have been
proposed by \cite{zhao2011reinforcement, goldberg2012q}. The latter
methods use Q-learning, relying on sequential support vector
regressions, to estimate the optimal sequential treatment rule that
optimizes a survival outcome under
right-censoring. \cite{geng2015optimal} also tackle estimation of the
optimal rule in a survival setting using $\ell_1$ regularization for
the outcome regression under the strong assumption that censoring is
independent of covariates and the outcome, but their decision
functions are restricted to linear functions. \cite{zhao2015doubly}
generalize the weighted classification approach of
\cite{zhao2012estimating} to allow for informative censoring and
doubly robust loss functions, but their decision functions are
restricted to support vector machines. \cite{bai2016optimal} present
methods for estimating optimal rules with a survival outcome subject
to informative censoring. They consider two strategies based on
estimation of the blip function and based on a classification
perspective. Their methods are restricted to decision functions that
can be indexed by a Euclidean vector and parametric nuisance
estimators, and are therefore of limited applicability to
high-dimensional data. All the above methods are potential candidates
in the library of estimators that constitute our ensembles.


In this article, we propose two methods to construct an ensemble of
decision functions for the optimal rule. Our ensembles are linear
combinations of estimators in a user-supplied library, where the
coefficients in the linear combination are chosen to minimize the
cross-validated risk. We propose to use a doubly robust loss function
with roots in efficient estimation theory for marginal causal effects
\citep{Moore2009,diaz2015improved}. In our context, double robustness
means that the estimated rules will have certain optimality properties
under consistent estimation of at least one of two nuisance
parameters: (a) the hazard of the outcome at each time point
conditional on covariates and treatment, and (b) the hazard of
censoring and the treatment mechanism.

The library of candidate estimators may contain any of the algorithms
discussed in the previous paragraphs. In light of the \textit{no free
  lunch} theorems of \cite{Wolpert2002} for supervised learning, for
any given dataset, our ensembles are expected to have better or equal
generalization error than any of the individual candidates in the
library. We provide a formal proof of this claim in the form of an
oracle inequality, which bounds the excess risk of our estimator in
terms of the excess risk of the oracle estimator, defined as the
combination of estimators that would be chosen in a hypothetical world
in which an infinite validation sample is available and at least one
of the nuisance parameters is known. Our methods are developed under
the assumption that censoring is at random \citep{rubin1987multiple},
which means that censoring is random within strata of treatment and
baseline variables. We also assume that treatment is randomized within
strata of the covariates, either by nature or by experimentation.

The finite sample bounds we present are inspired by developments in
the targeted learning literature, which establish the optimality of
cross-validation in estimator selection for high-dimensional
parameters \citep{vanderLaan&Dudoit03}. Related to our work,
\cite{luedtke2016super} consider super learning ensembles for
estimation of optimal DTRs in two time points. They present oracle
inequalities for super learning of the optimal rule using a loss
function indexed by the treatment mechanism, which is assumed known.
We generalize their results in the following ways: (i) we provide
oracle inequalities under a doubly robust loss function indexed by two
nuisance parameters, when neither of the nuisance parameters is known,
(ii) we show that the oracle inequalities inherit the double
robustness property of the loss function, and (iii) we present
comparable oracle inequalities for the 0-1 loss function. In addition,
we discuss how these oracle inequalities are related to the
convergence of the value of the rule under a margin assumption
describing the behavior of the blip function in the boundary of the
decision threshold.


\section{Data and Notation}\label{sec:notation}
Assume individuals are monitored at $K$ time points
$t=\{1,\dots,K\}$. Let $T$ denote a time-to-event outcome taking
values in $\{1,\dots,K\}\cup \{\infty\}$, where $T=\infty$ represents
no event occurring in the follow-up period.  Let $C \in \{0,\dots,K\}$
denote the censoring time defined as the time at which the individual
is last observed in the study, and let $C=K$, represent administrative
censoring.  Let $A\in\{0,1\}$ denote study arm assignment, and let $W$
denote a vector of baseline variables, which may include gene
expression as well as demographic, comorbidity, and other clinical
data. Denote $\one(\cdot)$ the indicator variable taking value $1$ if
the argument is true and $0$ otherwise.  The observed data vector for
each participant is $O=(W,A,\Delta, \tilde T)$, where
$\tilde T=\min(C,T)$, and $\Delta = \one\{T\leq C\}$ is the indicator
that the participant's event time is observed (uncensored). For a
random variable $X$, we let $X$ take values on a set $\mathbf O$.

We assume the observed data vector for each participant $i$, denoted
$O_i=(W_i,A_i,\Delta_i, \tilde T_i)$, is an independent, identically
distributed draw from the unknown joint distribution $P_0$ on
$(W,A,\Delta, \tilde T)$. The empirical distribution of
$O_1,\ldots,O_n$ is denoted with $\Pn$. We assume $P_0\in \mathcal M$,
where $\mathcal M$ is the nonparametric model defined as all
continuous densities on $O$ with respect to a dominating measure
$\nu$. We use $P$ to denote a generic distribution
$P \in \mathcal{M}$, and $E_0(\cdot)$ to denote expectation with
respect to $P_0$, and $\E(\cdot)$ is used to denote expectation over
draws of $O_1,\ldots, O_n$. For a function $f(o)$, we denote
$Pf=\int f(o)dP(o)$, and $||f||^2=P_0f^2$. We use $a\lesssim b$ to
denote that $a$ is smaller or equal than $b$ up to a universal
constant.

We can equivalently encode a single participant's data vector $O$
using the following longitudinal data structure:
\begin{equation}
  O=(W, A, R_0, L_1, R_1, L_2\ldots, R_{K-1}, L_K),\label{O}
\end{equation}
where $R_t = \one\{\tilde T = t, \Delta=0\}$ and
$L_t= \one\{\tilde T = t, \Delta=1\}$, for $t\in\{0,\ldots,K\}$. 
For a random variable $X$, we denote
its history through time $t$ as $\bar X_t=(X_0,\ldots,X_t)$. For a
given scalar $x$, the expression $\bar X_t=x$ denotes element-wise
equality. 

Define the following indicator variables for each $t \geq 1$:
$I_t=\one\{\bar R_{t-1}=0, \bar L_{t-1}=0\}$,
$J_t=\one\{\bar R_{t-1}=0, \bar L_t=0\}.$ The variable $I_t$ is the
indicator based on the data through time $t-1$ that a participant is
at risk of the event being observed at time $t$. 
Analogously, $J_t$
is the indicator based on the outcome data through time $t$ and
censoring data before time $t$ that a participant is at risk of
censoring at time $t$. We define $J_0=1$.

Define the discrete hazard function for survival at time $m \in \{1, \dots, K \}$:
\begin{equation}
  h(m,a,w)=P_0(L_m=1\mid I_m = 1, A=a, W=w),\nonumber
\end{equation}
among the population at risk at
time $m$ within strata of study arm and baseline variables. Similarly, for the censoring variable $C$, define the censoring hazard at time $m \in \{0, \dots, K \}$:
\begin{equation}
  g_R(m,a,w)=P_0(R_m=1\mid J_m=1, A=a, W=w).\nonumber
\end{equation}
We use the notation $g_A(a,w)=P_0(A=a\mid W=w)$, $g=(g_A,g_R)$, and $\eta=(h,g_A,g_R)$.
Let $p_{W}$ denote the marginal distribution of the baseline variables $W$.
We add the subscript $0$ to $p_W,g,h$ to denote the corresponding quantities under  $P_0$.

\section{Treatment Effect, Identification, and Optimal Individualized
  Treatment Rules}\label{sec:causal}

\subsection{Potential Outcomes and Causal Parameter}
Define the potential outcomes $T_a:a\in\{0,1\}$ as the event times
that would have been observed had study arm assignment $A=a$ and
censoring time $C=K$ been externally set with probability one. For a
restriction time $\tau\in\{1,\ldots,K\}$ of interest, we define the
restricted survival time under treatment arm $A=a$ as
$\min(T_a,\tau)$. For a transformation $Z$ of $W$, the treatment
effect within strata of the covariates $Z$ may be defined in terms of
the so-called full-data \textit{blip} function \citep[see e.g.,][]{robins1997causal} of the restricted
mean survival time:
\[\thetac(z) = E\{\min(T_1,\tau)-\min(T_0,\tau)\mid Z = z\}.\]
The transformation $Z$ may represent a subset of covariates (e.g.,
gene expression), or the whole vector $W$. We define the marginal
treatment effect as $\thetacm=E\{\min(T_1,\tau)-\min(T_0,\tau)\}$.

The subscript $c$ denotes a causal parameter, that is, a parameter
of the distribution of the potential outcomes $T_1$ and $T_0$. It can
be shown \citep[see][]{diaz2015improved} that $E\{\min(T_a,\tau)\mid Z=z\}=\sum_{t=0}^{\tau -
  1}\Sc(t,a,z),$ where $\Sc(t,a,z)=P(T_a>t\mid Z=z)$ is the survival
probability corresponding to the potential outcome under assignment to
arm $A=a$ within strata $Z=z$. As a result, $\thetac(z)$ may be
expressed as
\begin{equation}
  \thetac(z) = \sum_{t=1}^{\tau-1}\{\Sc(t,1,z) - \Sc(t,0,z)\}, \label{theta_c_def}
\end{equation}
since $\Sc(0,a,z)=1$ for $a\in \{0,1\}$ and for all $z$.

An individualized treatment rule $d$ is a function that maps the
covariate values $z$ of a given participant to a personalized
treatment decision in $\{0,1\}$. The potential time to event under a
rule $d$ is defined as $T_{d} = d(z)T_1 +
\{1-d(z)\}T_0$. Accordingly, the restricted mean survival time under a treatment rule that
assigns treatment according to $d(z)$ is equal to
\[E\{\min(T_{d},\tau)\}=E\{d(Z)[\min(T_1,\tau) - \min(T_0,\tau)]\}
  + E\{\min(T_0,\tau)\}.\] Because the last term does not depend on
$d(z)$, we define the \textit{value} of the rule $d$ as
\[\betac(d)=E\{d(Z)[\min(T_1,\tau) - \min(T_0,\tau)]\}=E\{d(Z)\thetac(Z)\}.\]
The above equation provides the basis for the definition of an
optimal rule as
\[\doptc(z) = \argmax_{d\in\mathcal D}\betac(d)=\one\{\thetac(z) > 0\},\]
where $\mathcal D=\{d:\mathbf Z \to \{0,1\}\}$ is the space of
functions that map the range of $Z$ into a treatment decision in
$\{0,1\}$. We define optimality of an rule with respect to the
restricted mean survival time, though other effect measures could also
be used.

\subsection{Identification of Parameters in Terms of Observed Data
  Generating Distribution $P_0$} \label{sec:identification}

In this section we show how the blip function $\thetac(z)$, the value
function $\betac(d)$, and the optimal rule $\doptc(z)$, which are
defined above in terms of the distribution of potential outcomes, can
be equivalently expressed as functions $\theta_0(z)$, $V_0(d)$, and
$\dopt(z)$ of the observed data distribution
$P_0(W,A,\Delta, \tilde T)$, under the assumptions \ref{ass:1}-\ref{ass:4} below.
This is useful since the potential outcomes are not always observed,
in contrast to the observed data vector $(W,A,\Delta, \tilde T)$ for
each participant, whose distribution we can make direct statistical
inferences about.

Define the following assumptions:
\begin{assumption}[Consistency]\label{ass:1}
  $T= \one(A=0) T_0 + \one(A=1) T_1$
\end{assumption}
\begin{assumption}[Randomization]\label{ass:2} $A$ is independent of $T_a$ conditional on $W$, for each
  $a\in\{0,1\}$
\end{assumption}
\begin{assumption}[Random censoring]\label{ass:3}
  $C$ is independent of $T_a$ conditional on $(A,W)$, for each
  $a\in\{0,1\}$
\end{assumption}
\begin{assumption}[Strong positivity]\label{ass:4}
  $P_0(g_{A,0}(a,W) > \epsilon)=1$ and $P_0(g_{R,0}(t,a,W) < 1-\epsilon)=1$ for each
  $a\in\{0,1\}$ and $t\in\{0,\ldots,\tau-1 \}$ and some $\epsilon>0$.
\end{assumption}
We make assumptions \ref{ass:1}-\ref{ass:4} throughout the manuscript. 
Denote the survival and censoring function for $T$ at time
$t \in \{1,\dots, \tau-1\}$ conditioned on study arm $a$ and baseline
variables $w$ by
\begin{equation*}
  S(t,a,w)=P(T>t\mid A=a,W=w), \quad G(t,a,w)=P(C\geq t\mid A=a,W=w).
\end{equation*}
Under assumptions \ref{ass:1}-\ref{ass:4}, we have $T \indep C \mid A,W$ and therefore
$S(t,a,w)$ and $G(t,a,w)$  have the following product formula representations:
\begin{align}
  S(t,a,w)&=\prod_{m=1}^t \{1-h(m,a,w)\}, \quad
            G(t,a,w)=\prod_{m=0}^{t-1} \{1-g_R(m,a,w)\}.
            \label{defS}
\end{align}
The potential outcome survival function $\Sc(t,a,z)$ can be equivalently
represented in terms of the observed data distribution as
$S(t,a,z)=E\{S(t,a,W)\mid Z=z\}$.
It follows from (\ref{theta_c_def}) that the causal parameter
$\thetac(z)$ is equal to the following observed-data blip function:
\begin{equation}
  \theta(z) =
  \sum_{t=1}^{\tau-1} E\left\{\prod_{m=1}^t \{1-h(m,1,W)\} -
    \prod_{m=1}^t \{1-h(m,0,W)\}\mid Z=z \right\}.
  \label{deftheta}
\end{equation}
Thus, the value $\betac(d)$ of a rule $d$ is equal to
$V(d)=E\{d(Z)\theta(Z)\}$, and a corresponding optimal treatment
rule is equal to $\dopt(z) = \one\{\theta_0(z)> 0\}$, where we denote
the corresponding true quantities (i.e., quantities computed
w.r.t. $P_0$) as $\theta_0(z)$, $V_0(d)$, and $\doptz(z)$.

In addition to assumptions \ref{ass:1}-\ref{ass:4} above, we sometimes
make the following margin assumption, which is common in the classification
literature for plug-in estimators:
\begin{assumption}[Margin assumption]\label{ass:ma}
  There exists a constant $\lambda\geq 0$ such that $P_0(0<
  \theta_0(Z)\leq t)\lesssim t^\lambda$ for all $t>0$.
\end{assumption}
The case $\lambda=0$ is trivial and implies no assumption, whereas
$\lambda=\infty$ corresponds to the strongest assumption since it
implies that $\theta_0(Z)$ is bounded away from zero. This assumption
characterizes the behavior of the decision function in the boundary,
and has been shown crucial to establish the convergence of certain
classifiers \cite[e.g.,][]{audibert2007fast,luedtke2017faster}.
\section{Plug-in Estimation of the Blip Function and the Optimal
  Rule}\label{sec:estimate}

In this section we discuss various estimators for $\theta_0(z)$, which
can be mapped to a plug-in estimators through
$\doptz(z) = \one\{\theta_0(z)> 0\}$. Our general strategy relies on
the concept of \textit{censoring unbiased transformation}, given in
Definition~\ref{def:cut} below. This concept was first introduced by
\cite{fan1994censored} and is further discussed in
\cite{rubin2007doubly}, among others.

\begin{definition}[Unbiased transformation]\label{def:cut}
  $D:\mathcal O\to \R$ is referred to as an unbiased transformation
  for $\theta_0(z)$ if
  $E_0\left\{D(O)\mid Z=z\right\} = \theta_0(z)$.
\end{definition}
The above definition motivates the construction of estimators of
$\theta_0(z)$ by regressing the transformation $D(O)$ on the
covariates $Z$. A common complication in this step is that most
unbiased transformations typically depend on unknown \textit{nuisance}
parameters which must be estimated prior to carrying out the
analysis. 
In this work, we focus on the \textit{doubly robust} censoring
unbiased transformation $D_\eta$ defined in Lemma~\ref{lemma:eif}
below. In addition to being a doubly robust unbiased transformation
for $\theta_0(z)$ (i.e., providing robustness to inconsistent
estimation of one out of two nuisance parameters), $D_\eta$ is an
efficient estimating function in the non-parametric model in the sense
that it may be used to construct efficient estimators of the marginal
treatment effect $\thetacm$ \cite[see e.g.,][]{diaz2015improved}.

\begin{lemma}[Doubly robust censoring unbiased
  transformation]\label{lemma:eif}
  Define
  \begin{equation} D_{\eta}(O)=
    \sum_{m=1}^{\tau-1}\big[I_mZ(m,A,W)\left\{L_m - h(m,A,W)\right\} +
    S(m,1,W) - S(m,0,W)\big],\label{defD}
  \end{equation}
  where $Z(m,A,W)=Z_1(m,A,W)-Z_0(m,A,W)$, and
  \begin{equation}
    Z_a(m,A,W)=-\sum_{t=m}^{\tau-1}\frac{\one\{A=a\}}{g_A(a,W)G(m,a,W)}
    \frac{S(t,a,W)}{S(m,a,W)}.\label{defZ}
  \end{equation}
  Assume $\eta=(h,g_A,g_R)$ is such that $h=h_0$ or
  $(g_A,g_R)=(g_{A,0}, g_{R,0})$. Then $D_\eta$ is an unbiased
  transformation for $\theta_0(z)$, that is,
  $E_0\left\{D_\eta(O)\mid Z=z\right\} = \theta_0(z)$ .
\end{lemma}
As a consequence of the previous lemma, the expected value of the
quadratic loss function $L_\eta(O;\theta) = \{D_\eta(O) - \theta(Z)\}^2$
is minimized at $\theta_0$ if $\eta=(h,g)$ is such that either $h=h_0$
or $g=g_0$.

For a loss function $L_\eta$, we denote its expected value as
$R_{0,\eta}(\theta)=E_0\{L_\eta(O;\theta)\}$ and refer to it as the
risk. We now discuss the construction of super learning ensembles of
candidate estimators for $\theta$ that target minimization of the
quadratic risk. Consider a collection of estimation algorithms for
estimating $\theta_0$, hereby called a \textit{library},
${\cal L}=\{\hat\theta_j:j=1,\ldots,J\}$. For an estimator $\hat\eta$
of $\eta_0$, in light of the discussion of the previous section, this
library may be constructed by considering any predictive algorithm
that minimizes the quadratic risk for prediction of the doubly robust
unbiased transformation $D_{\hat\eta}(O)$. The literature in machine
and statistical learning provides us with a wealth of algorithms that
may be used in this step. Examples include algorithms based on
regression trees (e.g., random forests, Bayesian regression trees),
algorithms based on smoothing (e.g., generalized additive models,
local polynomial regression, multivariate adaptive regression
splines), and others (e.g., support vector machines, neural networks).

Consider the following cross-validation set up.  Let
${\cal V}_1,\ldots,{\cal V}_K$ denote a random partition of the index
set $\{1,\ldots,n\}$ into $K$ validation sets of approximately the
same size. That is, ${\cal V}_k\subset \{1,\ldots,n\}$;
$\bigcup_{k=1}^K {\cal V}_k = \{1,\ldots,n\}$; and
${\cal V}_k\cap {\cal V}_{k'}=\emptyset$. In addition, for each $k$,
the associated training sample is
${\cal T}_k=\{1,\ldots,n\}\setminus {\cal V}_k$. Denote $\hat\eta_k$
the estimator of $\eta_0$ trained only using data in $\mathcal
T_k$. Likewise, denote by $\hat\theta_{j,k}$ the estimator of
$\theta_0$ obtained by training the $j$-th predictive algorithm in
$\mathcal L$ using only data in the sample ${\cal T}_k$ (e.g.,
regressing $D_{\hat\eta_k}(O_i)$ on $V_i$ for $i\in {\cal T}_k$). We
use $k(i)$ to denote the index of the validation set that contains
observation $i$.  The cross-validated prediction risk of
$\hat\theta_j$ is defined as
\[\hat R_{\hat\eta}(\hat\theta_j) =
  \frac{1}{K}\sum_{i=1}^n\frac{1}{|\mathcal V_{k(i)}|}
  L_{\hat\eta_{k(i)}}(O_i, \hat\theta_{j,k(i)}).\]
In this paper we
consider an ensemble learner given by a convex combination
\[\hat\theta_\alpha(z) = \sum_{j = 1}^{J}\alpha_j \hat\theta_j(z),\quad
  \alpha_j\geq 0,\quad \sum_{j=1}^{J}\alpha_j=1.\]
The weights $\alpha_j$ are chosen to minimize the
cross-validated risk of the above combination, that is:
\[\hat\alpha =
  \argmin_{\alpha} \sum_{i=1}^n\frac{1}{|\mathcal
    V_{k(i)}|}\left\{D_{\hat\eta_{k(i)}}(O_i) - \sum_{j = 1}^{J}
    \alpha_j \hat\theta_{j,k(i)}(V_i)\right\}^2\text{ subject to }
  \alpha_j\geq 0,\quad \sum_{j=1}^{J}\alpha_j=1.\] The above
expression is a weighted ordinary least squares problem with
constraints on the coefficients, and may therefore be solved using
standard off-the-shelf regression or optimization software. We denote
this super learner with $\thetaSL=\hat\theta_{\hat\alpha}$.

The optimality of general cross-validation selection procedures is
discussed in \cite{vanderLaan&Dudoit&vanderVaart06b,
  vanderVaart&Dudoit&vanderLaan06}. Optimality here is defined in
terms of asymptotic equivalence with the \textit{oracle} risk, which
we define as the risk computed when (i) one of the components of the
nuisance parameter $\eta_0$ is known, and (ii) a validation sample of
infinite size is available to assess the performance of the
estimator. Specifically,
\begin{definition}[Oracle risk and oracle selector] Let $\eta_1=(g_1, h_1)$, where either
  $g_1=g_0$, or $h_1=h_0$. The oracle risk of a candidate
  $\hat\theta_\alpha$ is defined as
  \[\tilde R_{\eta_1}(\hat\theta_\alpha)=\frac{1}{K}\sum_{k=1}^K\int
    \left\{D_{\eta_1}(o) - \hat\theta_{\alpha,k}(z)\right\}^2dP_0(o).\]
  The oracle selector is equal to
  \[\tilde\alpha= \argmin_\alpha \tilde R_{\eta_1}(\hat\theta_\alpha)\,\,\text{ subject
      to } \alpha_j\geq 0,\quad \sum_{j=1}^{J}\alpha_j=1,\]
  and the corresponding oracle blip function is denoted with $\thetaOR=\hat\theta_{\tilde\alpha}$.
\end{definition}

The risk $\tilde R_{\eta_1}(\theta_0)=\int L_{\eta_1}(o;\theta_0)dP_0(o)$ is
the optimal risk (with respect to the loss function $L_{\eta_1}$,
which in light of Lemma~\ref{lemma:eif} is a valid loss function)
achieved by the true $\theta_0$.  The following theorem provides a
bound on the excess risk of the estimator $\thetaSL$ and the excess
risk of $\thetaOR$. The excess risk for a selector $\hat\alpha$ is
defined as the difference between the oracle risk of the selector
$\hat\alpha$ and the optimal risk, i.e.,
\[\err^2(\hat \theta_{\hat\alpha}) = \E\{\tilde R_{\eta_1}(\hat \theta_{\hat\alpha})- \tilde
  R_{\eta_1}(\theta_0)\}=\E\,P_0(\hat\theta_\alpha -\theta_0)^2,\]
where we remind the reader that the expectation is taken over draws of
$O_1,\ldots, O_n$.  We denote this excess risk as
$\err^2(\hat \theta_{\hat\alpha})$, below we refer to its square root
as $\err(\hat \theta_{\hat\alpha})$. We show that the above excess
risk is bounded by two terms: one depending on the excess risk of the
oracle selector $\err(\thetaOR)$, and another one depending on doubly
robust terms associated to estimation of $\eta_0$.

\begin{theorem}[Oracle inequality for the super learner of the blip
  function]\label{theo:oracle}
  Let $\eta_1=(g_1,h_1)$ denote the element-wise $L_2(P_0)$ limit of
  $\hat\eta$ as $n\to\infty$, and assume that either $g_1=g_0$ or
  $h_1=h_0$. Define
  \begin{align*}
    B_1(\hat\eta, \eta_0)&=\E||(\hat g- g_0)(\hat h - h_0)||\\
    B_2(\hat\eta, \eta_0)&=\E\left\{\one(g_1=g_0)||\hat g- g_0|| + \one(h_1=h_0)||\hat h - h_0||\right\}^2.
  \end{align*}
  Then,
  for $\delta > 0$
  \begin{multline}\err(\thetaSL)\leq
    (1+2\delta)^{1/2}\err(\thetaOR) +
    C_1\left\{(1+\log n)/n\right\}^{1/2} + \\ C_2B_1(\hat\eta,\eta_0)+ C_3\left(\log
      n/n\right)^{1/4}\left\{B_2(\hat\eta,\eta_0)\right\}^{1/2}\label{eq:or}
  \end{multline}
  for constants $C_1$, $C_2$, and $C_3$.
\end{theorem}

Note that the terms $B_1(\hat\eta, \eta_0)$ and
$B_2(\hat\eta, \eta_0)$ converge to zero if either $\hat g$ or
$\hat h$ converge to $g_0$ or $h_0$, respectively, in $L_2(P_0)$
norm. This implies that the doubly robust property of $D_\eta(o)$ is
transferred to the oracle inequality. To the best of our knowledge
this result had not been previously shown in the literature.

The super learner $\thetaSL$ may be used to construct a plug-in
estimator of the optimal rule as $\hat d(z)=\one\{\thetaSL(z) >
0\}$. 
The following remark discusses the convergence
rates of the value of the selected rule $V_0(\hat d)$ to the value of
the oracle rule $V_0(\tilde d)$.
\begin{remark}[Convergence rates to the oracle value]\label{remark:rates}
  Assume
  \[B_1(\hat\eta, \eta_0)=O\left((\log n/n)^{1/2}\right),\quad
    B_2(\hat\eta, \eta_0)=O\left((\log n/n)^{1/2}\right).\] Lemma 5.3 of
  \cite{audibert2007fast}, along with Jensen's inequality, show that
  under assumption \ref{ass:ma}, we have
  \[\E\{V_0(\tilde d)-V_0(\hat d)\}\lesssim \{\err^2(\thetaSL) -
    \err^2(\thetaOR)\}^{(1+\lambda)/(2+\lambda)},\] where
  $\tilde d(z)=\one\{\thetaOR(z) > 0\}$ is the oracle rule.  This
  yields the following convergence rate:
  \[\E\{V_0(\tilde d)-V_0(\hat d)\}=O\left((\log
      n/n)^{(1+\lambda)/(2+\lambda)}\right).\] An example of a case
  yielding the above rate is a randomized study
  ($g_{A,0}(w)=q\in(0,1)$) with no censoring ($P_0(\Delta=1)=1$.) In
  this case, a logistic regression fit of $A$ on $W$ containing at least
  an intercept would yield an estimator satisfying
  $||\hat g_A - g_{A,0}||^2=O_P(n^{-1})$. Plugging in the true value
  $g_{R,0}(t,a,w)=0$ for $\hat g_R(t,a,w)$, and assuming $\hat h$ is
  inconsistently estimated yields $B_1(\hat\eta, \eta_0)=O(n^{-1/2})$
  and $B_2(\hat\eta, \eta_0)=O(n^{-1/2})$.
  Under no margin assumption ($\lambda=0$) we get a convergence rate of
  $(\log n/n)^{1/2}$. Under a strong margin assumption in which
  $\theta_0(Z)$ is bounded away from zero ($\lambda = \infty$) we get a
  rate of $\log n/n$.

  The above convergence result establishes the convergence
  of the value of our estimator $\hat d$ to the value of the oracle
  $\tilde d$. This is different from the typical result in the
  classification literature, which establishes convergence to the
  optimal value $V_0(d_0)$. The latter result often involves fast
  learning rates (sometimes faster than $n^{-1}$) and requires restricting
  the class of blip functions considered to H\"older
  \citep{audibert2007fast} or Donsker \citep{luedtke2017faster}
  classes, a restriction we do not impose.

\end{remark}

\section{Super Learner Ensembles for the Optimal Rule from a
  Classification Perspective}\label{sec:sl}

\subsection{Estimators Using the 0-1 Loss Function}
In this section we discuss a classification approach that aims at
directly estimating the optimal rule $\doptz(z)$.  Our approach here
differs from the previous section in that we do not attempt to
estimate the blip function. Instead, we introduce the concept of a
decision function, defined as $f:\mathbf Z\to \R$, and which yields a
treatment rule $d_f(z)=\one\{f(z) > 0\}$. In a slight abuse of
notation we use $V(f)$ to refer to the value of the rule $d_f$. Any
function $f_0$ such that $\sign\{f_0(z)\theta_0(z)\}=1$ has optimal
value $V_0(d_0)$. This provides intuition on the benefits of directly
optimizing the value of the loss function instead of the risk of the
blip function: an inconsistent estimator of the blip function may
provide an optimal rule, as long as its sign is correct. For a given
rule $d_f$, in light of Lemma~\ref{lemma:eif}, we have that
$V_0(f) = E_0\{d_f(Z)D_\eta(O)\}$ if $\eta$ is such that either
$h=h_0$, or $g=g_0$. Thus, a decision function that optimizes the
value of the rule $d_f$ may be found as
\[f_0\in\argmax_f \int d_f(z)D_\eta(o)d P_0(o).\] For a binary value
$b\in\{0,1\}$ and any $X$ we have
$bX = \one\{X>0\}|X| - |X|\one\left[\one\{X>0\} \neq b\right]$. Thus,
the optimization problem may be recast as $f_0\in \mathcal F_0$, where
$\mathcal F_0 = \argmin_f \int L_\eta(o;f) dP_0$ and
\begin{equation}
  L_\eta(o;f)=  |D_\eta(o)|\,\one\left[\one\{D_\eta(o) > 0\}
    \neq d_f(z)\right].\label{eq:zoloss}
\end{equation}
Expression (\ref{eq:zoloss}) is a weighted classification loss
function in which we aim to classify the binary outcome
$\one\{D_\eta(O) > 0\}$ based on data $Z$, using the 0-1 loss function
with weights given by $|D_\eta(O)|$. The objective is to classify an
individual who benefits from treatment arm $A=1$ (i.e., an individual
with $D_\eta(O) > 0$) as requiring treatment (i.e., $d_f(Z)=1$), while
penalizing for the loss $|D_\eta(O)|$ incurred if the individual were
misclassified.

In what follows we consider a library of algorithms for estimation of
the decision function $\mathcal L=\{\hat f_j(z):j,\ldots,J\}$. In
light of the discussion of the previous sections, the most natural
choice for a decision function is the blip function $\hat
\theta(z)$. However, we do not restrict our setup to functions with a
blip interpretation. In addition to estimators of the blip function
$\theta_0(z)$, the library may contain other decision functions such
as the support vector machines proposed by \cite{zhao2015doubly} and
the parametric decision functions of \cite{bai2016optimal}.

We construct an ensemble of the decision functions as
\begin{equation}
  \hat f_\alpha(z) = \sum_{j = 1}^{J}\alpha_j \hat f_j(z),\quad
  \alpha_j\geq 0.\label{eq:fens}
\end{equation}
In this way, we generate an ensemble optimal rule as
$\hat d_\alpha(z) = \one\{\hat f_{\alpha}(z)> 0\}$.  As in the
previous section, we define the super learner selector as
\[\hat\alpha \in
  \argmin_{\alpha} \sum_{i=1}^n\frac{1}{|\mathcal V_{k(i)}|}
  L_{\hat\eta_k(i)}\left(O_i, \hat f_{\alpha,k(i)}\right)\text{
    subject to } \alpha_j\geq 0,\] where $\hat f_{\alpha,k(i)}$ represents
(\ref{eq:fens}) with $\hat f_j(z)$ replaced by $\hat f_{j,k(i)}(z)$:
the $j$-th decision function estimated using the training sample
$\mathcal T_{k(i)}$. The super learner of the decision function is
defined as $\fSL(z)=\hat f_{\hat\alpha}(z)$, and the corresponding
optimal rule is defined as $\dSL=\one\{\fSL(z)> 0\}$.

For $\eta_1=(g_1,h_1)$ such that either $g_1=g_0$ or $h_1=h_0$, the
oracle risk of the decision function is defined as
\[\tilde R_{\eta_1}(\hat f)=\frac{1}{K}\sum_{k=1}^K\int L_{\eta_1}(o,
  \hat f_k) dP_0(o).\] The oracle selector of $\alpha$ is thus defined
as $\tilde\alpha \in \argmin_{\alpha}\tilde R_{\eta_1}(\hat f_\alpha)$,
and we denote $\fOR=\hat f_{\tilde\alpha}$. The excess risk of an
estimator $\hat f$ is equal to
\[\err(\hat f)=\E\{\tilde R_{\eta_1}(\hat f)-\tilde R_{\eta_1}(f_0)\} =
  V_0(f_0) -\E V_0(\hat f).\]

In Theorem \ref{theo:oraclezo} below, we provide bounds on
$\err(\fSL)$ in terms of the excess risk of the oracle selector
$\err(\fOR)$ and the bias terms $B_1(\hat
\eta,\eta_0)$ and $B_2(\hat \eta,\eta_0)$ defined in
Theorem~\ref{theo:oracle}.
\begin{theorem}[Oracle inequality for the super learner of the optimal
  rule]\label{theo:oraclezo}
  Assume the conditions of Theorem~\ref{theo:oracle}. In addition,
  assume that $\hat\alpha$ is computed in a grid of size $Mn^q$ for
  some $M>0$, $q>0$. Then,
  \[0\leq \err(\fSL) \leq \err(\fOR) + C_1(\log n/n)^{1/2} + C_2
    B_1(\hat\eta,\eta_0).\] If condition \ref{ass:ma} holds with
  $\lambda=\infty$, then, for $\delta > 0$
  \[0\leq \err(\fSL) \leq
    (1+2\delta)\err(\fOR) +
    C_1\frac{1+\log n}{n} + C_2 B_1(\hat\eta,\eta_0) + C_3\sqrt{\frac{\log
        n}{n}}B_2(\hat\eta,\eta_0).\]
  for constants $C_1$, $C_2$, and $C_3$, where $B_1$ and $B_2$ are
  defined as in Theorem~\ref{theo:oracle}.
\end{theorem}

\begin{remark} Assume $B_1(\hat\eta,\eta_0)$ converges as in
  Remark~\ref{remark:rates}. An immediate consequence of the above
  result is that under no margin assumption, we have
  $\E\{V_0(\dOR)-V_0(\dSL)\}=O\left((\log n/n)^{1/2}\right)$. Under
  the strong margin assumption \ref{ass:ma} with $\lambda=\infty$ we
  have $\E\{V_0(\dOR)-V_0(\dSL)\}=O(\log n/n)$. These rates are
  identical to the rates obtained in Remark~\ref{remark:rates} for the
  plug-in estimator. The question of whether analogous convergence
  rates may be obtained for other values of $\lambda$ under condition
  \ref{ass:ma} remains an open problem.
\end{remark}

In comparison to Theorem~\ref{theo:oracle},
Theorem~\ref{theo:oraclezo} has the additional assumption that the
optimization of the loss function is carried out in a grid polynomial
size in $n$. Inspection of the proofs of the theorems in the
Supplementary Material reveals the reason for the additional
assumption: the 0-1 loss function is non-smooth and the Lipschitz
condition used in the proof of Theorem~\ref{theo:oracle} does not
apply. As demonstrated in our data application, this assumption is
likely to have little practical consequences, but it is unclear to us
whether it can be removed.

\subsection{Using a Surrogate Loss Function for the 0-1 Loss}\label{sec:su}
It is well known in the statistical learning literature that
minimizing (\ref{eq:zoloss}) is generally difficult due to the
discontinuity and non-convexity of the 0-1 loss. A common approach to
mitigate the issues arising from the discontinuity and non-convexity
of the 0-1 loss function is to use surrogates loss functions, such as
the logistic loss $\phi(x)=\log(1+\exp(-x))$ or the hinge loss
$\phi(x)=\max(1-x,0)$. We have the following result, which teaches us
that any decision function $d_{f_0}(z)$ based on a decision function
$f_0\in \mathcal F_{\phi,0}$ has the same performance as the optimal
rule $\doptz(z)$.
\begin{lemma}\label{lemma:suloss}
  Assume $\eta$ is such that either $h=h_0$, or
  $g=g_0$. Define
  \begin{equation}
    \mathcal F_{\phi,0}=\argmin_f \int L_{\phi,\eta}(o;f) dP_0(o),\label{eq:suloss}
  \end{equation}
  where the surrogate loss $L_{\phi,\eta}$ is defined as
  \[L_{\phi,\eta}(o;f)=|D_\eta(o)|\,\phi\big( f(z)\left[2
      \one\{D_\eta(o) > 0\} - 1\right]\big).\] Then we have
  $\mathcal F_{\phi,0}\subseteq \mathcal F_0$, where
  $\mathcal F_0 = \argmin_f \int L_\eta(o;f) dP_0$.
\end{lemma}

\section{Estimating the Optimal Treatment for Breast
  Cancer Patients in our Motivating Application}\label{sec:aplica}

Different types of human breast cancer tumors have been shown to have
heterogeneous response to treatments \citep{perou2000, sotiriou2009}.
Amplification of ERBB2 gene and associated overexpression of human
epidermal growth factor receptor (HER2) encoded by this gene occur in
25-30$\%$ of breast cancers \citep{slamon2001use}. HER2-positive
breast cancer is an aggressive form of the disease and the prognosis
for such patients is generally poor \citep{slamon1987human,
  seshadri1993}. The clinical efficacy of adjuvant trastuzumab, a
recombinant monoclonal antibody, in early stage HER2-positive patients
was demonstrated by several large clinical trials
\citep{perez2011four, romond2005}. Despite significant improvement in
disease-free and overall survival of patients treated with
trastuzumab, about 20-25$\%$ patients relapse within 3-5 years
\citep{perez2011four}. In this paper we use data from the North
Central Cancer Treatment Group N9831 study, a phase III randomized
clinical trial testing the addition of trastuzumab to chemotherapy in
stage I-III HER2-positive breast cancer.

The total number of patients enrolled in the NCCTG N9831 trial was
3,505. Samples from 1,390 patients, for whom there was available
tissue, were used to quantify mRNA from a custom codeset of 730 genes
created by experts. The available baseline variables may be thus be
categorized in three classes: demographic (e.g, race, age, ethnicity),
clinical (e.g., tumor grade, tumors size, nodal status, hormone
receptor status), and gene expression. Among the 1,390 patients, 483
received chemotherapy alone (control arm) and 907 patients received
chemotherapy plus trastuzumab (treatment arm).

The clinical challenge is to identify genetic and demographic profiles
for patients with HER2-positive breast cancer who are unlikely to
benefit from adjuvant trastuzumab.

In order to estimate and assess the performance of the estimated rule
using different datasets, we split our data into training and
validation datasets, of sizes 1000 and 390, respectively.
\subsection{Estimators of $h$ and $g_R$}\label{sec:slhg}
According to our theoretical results, the optimality of the estimated
treatment rules hinges upon consistent estimation of at least one of
the nuisance parameters $h$ and $g_R$. As a result, it is crucial to
employ flexible methods capable of unveiling complex patterns which
are not visible to the human eye. As demonstrated below in
Section~\ref{sec:assess}, simple parsimonious models such as the Cox
proportional hazards or logistic regression fail to detect these
complex relations in the data.

In order to accurately estimate the nuisance parameters, we use an
ensemble learner known as the super learner for prediction
\citep{vanderLaan&Polley&Hubbard07}. We train the ensemble separately
using data from each treatment arm, in order to fully account for
treatment-covariate interactions. Like our rule ensembles, super
learning predictors build a combination of candidate predictors that
minimize a cross-validated user-supplied risk function. Since $g_R$
and $h$ are conditional probabilities, we focus on logistic regression
ensembles and the negative log-likelihood loss function, using the R
implementation in the SuperLearner package \citep{SL}. The candidate
estimators included in the ensembles are listed in
Table~\ref{tab:coefshg}, along with the coefficients
of each predictor in the ensemble, when trained in the complete
dataset using 5-fold cross-validation.

\begin{table}[!htb]
    \begin{tabular}{rrrrrrrr}
      \hline
      & &  RF & XGB & MLP & GLM & MARS & LASSO \\
      \hline
      \multirow{2}{*}{$\hat g_R$} & $A=1$ & 0.100 & 0.056 & 0.000 & 0.000 & 0.312 & 0.532 \\
      & $A=0$ & 0.000 & 0.301 & 0.000 & 0.000 & 0.294 & 0.405 \\\hline
      \multirow{2}{*}{$\hat h$}   & $A=1$ & 0.023 & 0.050 & 0.000 & 0.000 & 0.237 & 0.691 \\
      & $A=0$ & 0.154 & 0.086 & 0.098 & 0.060 & 0.123 & 0.480 \\
      \hline
    \end{tabular}
    \caption{Coefficients of the super learner ensemble for estimation of
      $g_R$ and $h$. RF is random forests, XGB is extreme gradient
      boosting, MLP is multilayer perceptron, GLM is logistic
      regression, MARS is multivariate adaptive splines, and LASSO is $L_1$ regularized logistic regression.}
  \label{tab:coefshg}
\end{table}

For random forests, extreme gradient boosting, and multilayer
perceptron, the tuning parameters are tuned using data splitting with
the aid of the R caret package \citep{caret}. To avoid $p>n$, logistic
regression and multivariate adaptive splies are estimated with a
variable screening algorithm which computes univariate t-statistics
and keeps only the 50 variables with a larger value.

\subsection{Candidate Estimators for the Optimal Treatment Rule}

According to our discussion in Sections~\ref{sec:estimate} and
\ref{sec:sl}, there are at least three types of estimators for the
optimal rule $\dopt$. The first type is a simple substitution
estimator, obtained through inspection of equation~(\ref{deftheta}),
which consists in regressing the blip function
$\hat B(W)=\sum_{t=1}^{\tau-1}\{\hat S(t, 1, W) - \hat S(t,0,W)\}$ on
$Z$, where $\hat S$ is the estimator of the survival function
corresponding to the estimator $\hat h$ described in
Section~\ref{sec:slhg}. The second type is obtained through regression
of the unbiased transformation $D_{\hat\eta}(O)$ on $Z$.  The third
type of estimation methods is obtained based on
equation~(\ref{eq:zoloss}), and is obtained by classifying the binary
outcome $\one\{D_{\hat \eta}(O)>0\}$ as a function of $Z$, with
weights given by $|D_{\hat\eta}(O)|$. Here, $\hat\eta=(\hat g_R, \hat h)$,
where the components of $\hat\eta$ are as described in
Section~\ref{sec:slhg}. Any regression or supervised classification
technique available in the statistical learning literature may be used
as a candidate for solving these problems.

In our application, we focus on the following candidates for
estimating $\dopt$:

\vspace{.5cm}

\begin{tabular}{lp{11.5cm}}
  \textsc{B-Reg} &Regression of the blip function $\hat B(W)$ using super learning with
                   candidate learners as described in Table~\ref{tab:coefshg}.\\[0.2cm]
  \textsc{D-Reg} &Regression of the doubly robust transformation $D_{\hat\eta}(O)$ using super learning with
                   candidate learners as described in Table~\ref{tab:coefshg}.\\[0.2cm]
  \textsc{D-Class-RF} &Weighted classification of $\one\{D_{\hat \eta}(O)>0\}$
                        using random forests.\\[0.2cm]
  \textsc{D-Class-XGB} &Weighted classification of $\one\{D_{\hat \eta}(O)>0\}$
                         using extreme gradient boosting.\\[0.2cm]
  \textsc{D-Class-GLM} &Weighted classification of $\one\{D_{\hat \eta}(O)>0\}$
                         using logistic regression.\\
\end{tabular}

\vspace{.5cm}

According to our discussion in Sections~\ref{sec:estimate} and
\ref{sec:sl}, we also train four super learning ensembles of the above
candidate estimators, using different loss functions:

\vspace{.5cm}

\begin{tabular}{lp{11cm}}
  \textsc{SL-Reg}         & Regression ensemble minimizing the
                            expected quadratic loss function. \\[0.2cm]
  \textsc{SL-Class-01}    & Classification ensemble minimizing the expected 0-1 loss function.                               \\[0.2cm]
  \textsc{SL-Class-Hinge} & Classification ensemble with surrogate hinge loss function.                                      \\[0.2cm]
  \textsc{SL-Class-Log}   & Classification ensemble with surrogate log loss function.
\end{tabular}

\vspace{.5cm}

The coefficients of each candidate estimator in each ensemble are
presented in Table~\ref{tab:coefs}. These coefficients were computed using the
Subplex \citep{Rowan1990} routine implemented in the NLopt
nonlinear-optimization R package. For improved robustness, the 0-1
loss was optimized using 1000 different random starting values.

\begin{table}[!htb]
    \begin{tabular}{rrrrr}
      \hline
      & \textsc{SL-Reg} & \textsc{SL-Class-0-1} & \textsc{SL-Class-Hinge} & \textsc{SL-Class-Log} \\
      \hline
      \textsc{D-Class-RF} & 0.000 & 0.005 & 0.000 & 0.000 \\
      \textsc{D-Class-XGB} & 0.792 & 0.001 & 0.945 & 0.869 \\
      \textsc{D-Class-GLM} & 0.000 & 0.031 & 0.000 & 0.000 \\
      \textsc{D-Reg} & 0.007 & 0.017 & 0.001 & 0.006 \\
      \textsc{B-Reg} & 0.201 & 0.947 & 0.054 & 0.125 \\
      \hline
    \end{tabular}
    \caption{Coefficients of each candidate in each ensemble (standardized
      to sum one). The rows represent the candidates, the columns the
      ensemble.}
  \label{tab:coefs}
\end{table}

\subsection{Assessing the Performance of The Estimated Treatment Rule}\label{sec:assess}

Once each rule is estimated using only data in the training dataset,
its value $V(\hat d)$ is estimated on the validation dataset. To that
effect, we use the targeted minimum loss based estimator of the
restricted mean survival time proposed by \cite{diaz2015improved}
\citep[See also][]{Moore11}.

Figure~\ref{fig:values} presents the estimated restricted mean
survival time obtained with each estimated rule, along with 95\%
confidence intervals. For comparison, we also present the value of two
static rules of interest: never treat and always treat. As is clear
from the figure, the best algorithm in our application is regression
of the blip function. All super learning ensembles yield a similar
value, demonstrating the oracle property of the super
learner. Treating patients according to the optimal rule yields a
restricted mean survival of 157.1 (s.d. 3.1) months. In comparison
with the \textit{always treat} rule, which yields 151.2 (s.d., 3.3)
months, the optimal rule improves mean patient survival by 6 months.

According to Table~\ref{tab:coefs}, only the super learning ensemble
based on the 0-1 loss assigns a large weight to the best algorithm. In
fact, its restricted mean survival time (see Figure~\ref{fig:values})
is identical to that of the optimal rule. The other ensembles assign
more weight to the second best algorithm, weighted classification
using extreme gradient boosting, and have slightly smaller restricted
mean survival time. This is in agreement with our theoretical findings
that the best performance is obtained using the 0-1 loss function.

It is also worth noting that three of the estimated rules (weighted
classification using random forests and logistic regression, and
regression of the function $D_{\hat\eta}$) yield a restricted mean
survival time smaller or equal than the restricted mean survival time
of the static rule \textit{always treat}.
\begin{figure}[!hbt]
  \centering
  \includegraphics[scale = 0.3]{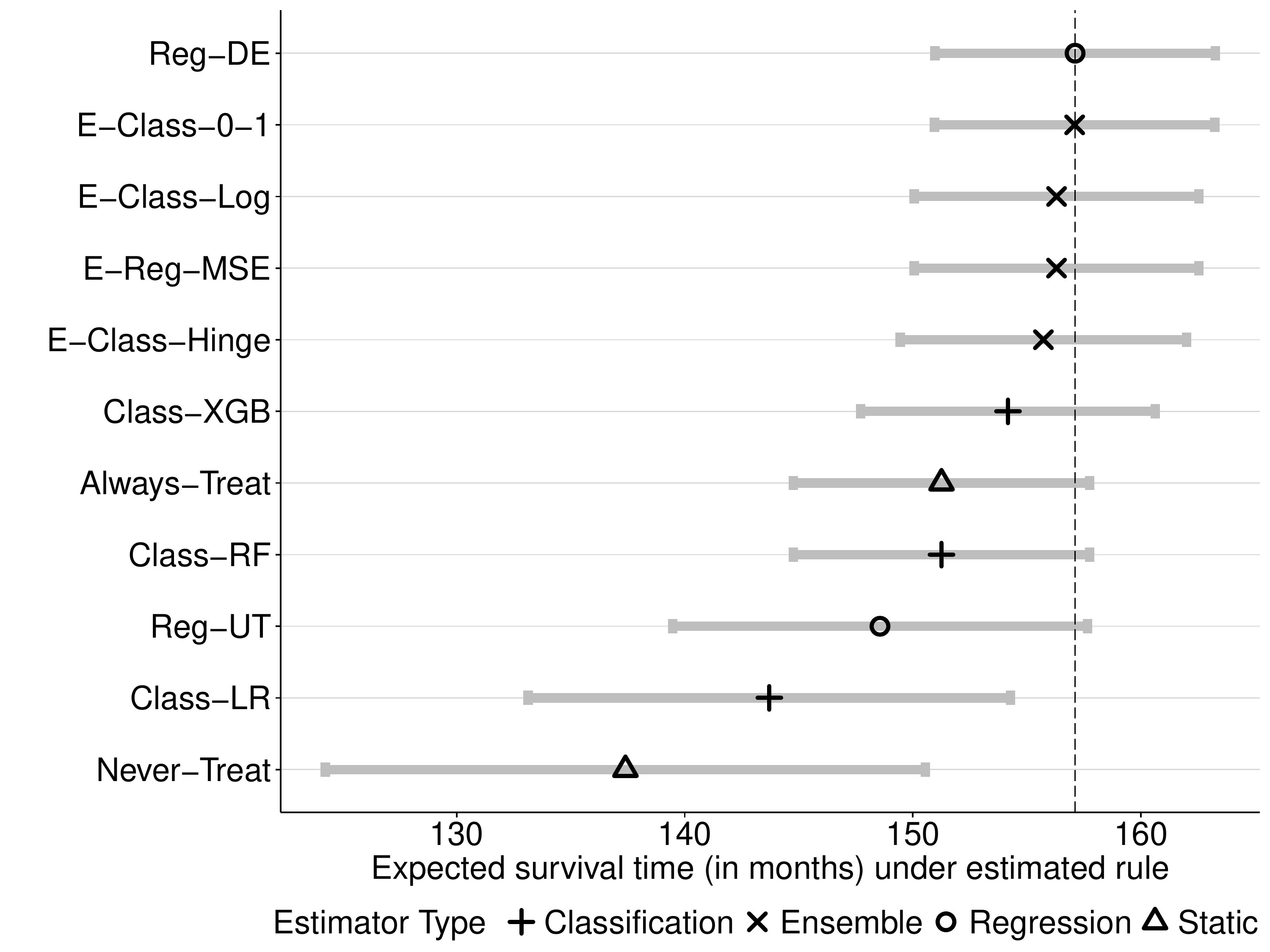}
  \caption{Restricted mean survival time estimated in the validation set, for different estimated
    rules. The bars represent 95\% confidence intervals. }\label{fig:values}
\end{figure}

Table~3 in the Supplementary Materials shows the p-value
for the pair-wise comparisons of the value of each estimated rule. A
few interesting points to note are:
\begin{enumerate}[(i)]
\item The ensemble using the 0-1 loss function outperforms the other
  ensembles, the difference in values is small but significant at 5\%
  level.
\item The value of the optimal rule, which is obtained through regression of the
  blip function (see Figure~\ref{fig:values}), is significantly
  different from all other rules, except the ensemble using the 0-1
  loss function. This illustrates the theoretical property of
  super-learning stating that the risk of the ensemble converges to
  the risk of the best candidate in the library.
\item Weighted classification using logistic regression, which is
  often advocated because it yields parsimonious rules
  \citep[e.g.,][]{zhang2015using}, has a value significantly lower than the
  static rule \textit{always treat}.
\end{enumerate}

In our application, we have decided to use data splitting to train and
assess the performance of the estimated rules. Though correct, this
approach may be unnecessary, since the value of the rule may be
assessed using the training dataset, under certain conditions derived
by \cite{luedtke2016statistical}.

\section{Discussion}

We present two methods for constructing an ensemble individualized
treatment rule. The methods are based on a plug-in estimator
optimizing the prediction error of the blip function, and a weighted
classification approach which directly estimates the decision
function. Though we found no theoretical differences between the two
approaches in terms of their asymptotic properties, the classification
ensemble using the 0-1 loss function yielded better treatment rules
than the other approaches in our illustrative application. The
superiority of the classification approach has been recognized
before \citep[e.g.,][]{zhao2012estimating}, and is a consequence of
the fact that it emphasizes optimizing the decision rule rather than
prediction accuracy emphasized by the blip approach.

We consider a survival time measured in a discrete time scale. Most
clinical research studies measure time to event in a discrete scale.
In our motivating application, time to relapse of cancer or death was
measured in days. We foresee no technical difficulties in extending
our approach to consider a continuous time to event. This can be
achieved by replacing discrete time hazards by their continuous
counterpart, as well as replacing certain sums over time by the
appropriate martingale integrals \cite[see][]{bai2016optimal} in the
definition of the censoring unbiased transformation $D_\eta(O)$. A
potential practical limitation is that the software and literature for
data-adaptive machine learning estimation of continuous time hazards
(required for the nuisance parameters) is scarce in comparison to that
of binary classification, which may be used for estimation of discrete
time hazards. Among the few methods that can be used for this problem
are (semi)-parametric models such as Cox regression and accelerated
failure time models. Available data adaptive approaches include
survival random forests and regularized Cox regression. If time is
measured on a continuous scale, implementation of our methods requires
discretization.  The specific choice of the discretization intervals
may be guided by what is clinically relevant.  For example, in cancer
research, the clinically relevant scale would typically be a week or a
month. In the absence of clinical criteria to guide the choice of
discretization level, a concern is that too coarse of a discretization
may lead to relevant information loss. A question for future research
is how to optimally set the level of discretization in order to trade
off information loss versus estimator precision. Another area for
future research is to consider discretization levels that get finer
with sample size.

We present doubly robust oracle inequalities and convergence rates
assuming (i) an exposure of interest that occurs at baseline, and (ii)
censoring which is confounded with the time to event only by baseline
variables. We conjecture that our general results apply to the more
general case of a dynamic treatment regime with a time-varying
treatment and time-varying confounders.  Such results will be the
subject of future research.

In our definition of oracle risk and oracle selector we have used a
nuisance parameter $\eta_1=(g_1,h_1)$ satisfying either $h_1=h_0$, or
$g_1=g_0$. Therefore these oracle quantities change depending on which
of the two nuisance parameters is correctly specified. According to
efficient estimation theory in semi-parametric models, we expect the
case $\eta_1=\eta_0$ to yield oracle quantities with minimal
variability. In the single misspecification case in which $h_1=h_0$ or
$g_1=g_0$ but not both, it is unclear to us whether misspecification
of one of the models yields better results than the other. Lastly, our
setup includes as particular case the inverse probability weighted
loss function, which may be obtained by using a constant estimator
$\hat h(t,a,w)=1$, as well as the g-computation loss function, which
is obtained by using $\hat g_A(a,w)=1$ and $\hat g_R(t,a,w)=0$.
\bibliographystyle{plainnat} \bibliography{tmle}

\section{Supplementary Material}
\subsection{Motivating application}

\begin{table}[!htb]
  \centering
  \caption{P-values of pair-wise comparisons of the value of each rule
    estimated in the validation data.}
  {\small    \begin{tabular}{rrrrrrrrrrrr}
               \hline
               \textsc{\small SL-Class-Log}   & 0.026 &       &       &       &       &       &       &       &       &       \\
               \textsc{\small SL-Class-Hinge} & 0.026 & 0.159 &       &       &       &       &       &       &       &       \\
               \textsc{\small SL-Reg}         & 0.025 & 0.160 & 0.443 &       &       &       &       &       &       &       \\
               \textsc{\small D-Class-RF}     & 0.001 & 0.001 & 0.001 & 0.001 &       &       &       &       &       &       \\
               \textsc{\small D-Class-XGB}    & 0.001 & 0.009 & 0.007 & 0.007 & 0.003 &       &       &       &       &       \\
               \textsc{\small D-Class-GLM}    & 0.001 & 0.001 & 0.001 & 0.001 & 0.011 & 0.001 &       &       &       &       \\
               \textsc{\small D-Reg}          & 0.001 & 0.004 & 0.002 & 0.002 & 0.158 & 0.021 & 0.056 &       &       &       \\
               \textsc{\small B-Reg}          & 0.120 & 0.025 & 0.023 & 0.023 & 0.001 & 0.001 & 0.001 & 0.001 &       &       \\
               \textsc{\small Always-Treat}   & 0.001 & 0.001 & 0.001 & 0.001 &       & 0.003 & 0.011 & 0.158 & 0.001 &       \\
               \textsc{\small Never-Treat}    & 0.001 & 0.001 & 0.001 & 0.001 & 0.001 & 0.001 & 0.036 & 0.002 & 0.001 & 0.001 \\
               \hline
                                              & \rot{\textsc{\small SL-Class-0-1}} &
                                                                                     \rot{\textsc{\small SL-Class-Log}} &\rot{\textsc{\small SL-Class-Hinge\,\,\,\,}}  & \rot{\textsc{\small SL-Reg}} & \rot{\textsc{\small D-Class-RF}} & \rot{\textsc{\small D-Class-XGB}} & \rot{\textsc{\small D-Class-GLM}}  & \rot{\textsc{\small D-Reg}}   & \rot{\textsc{\small B-Reg}} & \rot{\textsc{\small Always-Treat}}
             \end{tabular}}
           \label{tab:pvals}
         \end{table}

         \subsection{Proofs of Theorems and Lemmas}

         \subsubsection{Lemma~\ref{lemma:eif}}
         \begin{proof}
           For simplicity, consider the treatment-time-specific function
           \[D_{m,a,\eta}(O)=-\sum_{t=1}^m\frac{\one\{A=a\}I_t}{g_A(a,W)G(t,a,W)}\frac{S(m,a,W)}{S(t,a,W)}\{L_t-h(t,a,W)\}
             + S(m,a,W),\] and note that
           $D_\eta=\sum_{m=1}^{\tau-1}(D_{m,1,\eta}-D_{m,0,\eta})$. For a function $f(t,a,w)$ we denote
           $Pf(t)=\int f(t,a,w)dP(w)$. Conditioning first on $W$ in the above
           display yields
           \[E_0\{D_{m,a,\eta_0}\mid Z\}=E_0\left\{\prod_{t=1}^m\{1-h_0(t)\}\mid Z\right\}.\]
           Thus, we have
           \begin{align*}
             E_0(D&_{m,a,\eta}\mid Z) -
                    E_0\left\{\prod_{t=1}^m\{1-h_0(t)\}\mid Z\right\}\\
             =&E_0\left[\sum_{t=1}^m-\frac{S(m)}{S(t)}\frac{g_{A,0}}{g_A}\frac{G_0(t)}{G(t)}S_0(t)\{h_0(t)-h(t)\}
                + \prod_{t=1}^m\{1-h(t)\} - \prod_{t=1}^m\{1-h_0(t)\}\,\bigg
                | Z\right]\\
             =&\sum_{t=1}^mE_0\left[-\frac{S(m)}{S(t)}\frac{g_{A,0}}{g_A}\frac{G_0(t)}{G(t)}S_0(t-1)\{h_0(t)-h(t)\}
                + S_0(t-1)\{h_0(t)-h(t)\}\frac{S(m)}{S(t)}\,\bigg
                | Z\right]\\
             =&\sum_{t=1}^mE_0\left[-\frac{S(m)}{S(t)}S_0(t-1)\{h_0(t)-h(t)\}\left\{\frac{g_{A,0}}{g_A}\frac{G_0(t)}{G(t)}-1\right\}\,\bigg
                | Z
                \right]\\
             =&\sum_{t=1}^mE_0\left[-\frac{S(m)}{S(t)}S_0(t-1)\{h_0(t)-h(t)\}\left\{\frac{g_{A,0}}{g_AG(t)}\{G_0(t)-G(t)\}
                + \frac{1}{g_A}(g_{A,0}-g_A)\right\}\,\bigg
                | Z
                \right]\\
             =&\sum_{t=1}^mE_0\left[-\frac{S(m)}{S(t)}S_0(t-1)\{h_0(t)-h(t)\}\left\{\frac{g_{A,0}}{g_AG(t)}\sum_{k=0}^{t-1}G_0(k)\{g_{R,0}(k)
                - g_R(k)\}\frac{G(t)}{G(k+1)}\right.\right.\\
                  &\left.\left.
                    + \frac{1}{g_A}(g_{A,0}-g_A)\right\}\,\bigg
                    | Z
                    \right]
           \end{align*}
           Plugging in $g=g_0$ or $h=h_0$ yields the result.
         \end{proof}
         \subsubsection{Theorem \ref{theo:oracle}}
         \begin{proof}
           We start by assuming the minimization of the risk in the definition
           of $\hat\alpha$ and $\tilde\alpha$ is carried out in a grid
           $\mathcal B_n \subset \mathcal B=\{\alpha \in \R^J:\alpha_j\geq 0,
           \sum_{j=1}^{J}\alpha_j=1\}$ of polynomial size in $n$ (that is
           $|\mathcal B_n|\lesssim n^q$) for some $1 \leq q<\infty$, but do away
           with this assumption at the end of the proof. Let $\hat\beta$ ad
           $\tilde\beta$ denote the cross-validated and oracle selectors when
           the risk minimization is performed in $\mathcal B_n$ rather than
           $\mathcal B$. We use $P_{n,k}$ to denote the empirical distribution
           corresponding to the validation set $\mathcal V_k$, as well as
           $E_K(X)=K^{-1}\sum_{k=1}^KX_k$ to denote an average across
           validation splits. We denote $PL(\theta) = \int L(o;\theta)
           dP(o)$. Let
           \[\eta^\star =
             \begin{cases}
               (g_0,\hat h_k)&\text{if } g_1 = g_0\text{ and } h_1 \neq h_0\\
               (\hat g_k, h_0)&\text{if } g_1 \neq g_0\text{ and } h_1 = h_0\\
               (g_0, h_0)&\text{if } g_1 = g_0 \text{ and } h_1 = h_0
             \end{cases}\]
           Define the centered loss function
           \[L^0_\eta(O;\theta)=L_\eta(O;\theta) - L_{\eta_1}(O;\theta_0).\] In
           this proof we denote with $R$ the corresponding centered risks, i.e.,
           denote
           \begin{eqnarray*}
             \hat R_{\hat\eta}(\hat\theta)& = & \frac{1}{K}\sum_{k=1}^K\frac{1}{|{\cal V}_k|}\sum_{i\in
                                                {\cal V}_k}L^0_{\hat\eta_k}\left(O_i;
                                                \hat\theta_k\right)\\
             \tilde R_{\hat\eta}(\hat\theta)&=&\frac{1}{K}\sum_{k=1}^K\int L^0_{\hat\eta_k}(o;\hat\theta_k)dP_0(o)
           \end{eqnarray*}
           the corresponding  cross-validated and oracle risks.  For notational convenience we denote
           $R(\beta)=R(\hat\theta_\beta)$.
           Note that $\tilde R_{\eta^\star}(\beta) = \err(\hat\theta_{\beta})$. For $\delta > 0$ we have
           \begin{align}
             0\leq&\, \tilde R_{\eta^\star}(\hat\beta)\notag\\
             \leq&\,\tilde R_{\eta^\star}(\hat\beta)+(1+\delta)\{\hat
                   R_{\hat\eta}(\tilde\beta)-\hat R_{\hat\eta}(\hat\beta)\}\notag\\
             =&\,(1+2\delta)\tilde R_{\eta^\star}(\tilde\beta)\notag\\
                  &-(1+\delta)\{\hat R_{\eta^\star}(\hat\beta) - \tilde R_{\eta^\star}(\hat\beta)\}
                    -\delta \tilde R_{\eta^\star}(\hat\beta)\label{eq:T}\\
                  &+(1+\delta)\{\hat R_{\eta^\star}(\tilde\beta) - \tilde R_{\eta^\star}(\tilde\beta)\}
                    -\delta \tilde R_{\eta^\star}(\tilde\beta)\label{eq:R}\\
                  &+(1+\delta)\{\hat R_{\hat\eta}(\tilde\beta) - \hat
                    R_{\eta^\star}(\tilde\beta)\}\label{eq:A}\\
                  &-(1+\delta)\{\hat R_{\hat\eta}(\hat\beta) - \hat R_{\eta^\star}(\hat\beta)\}\notag
           \end{align}
           where the second inequality is a consequence of the definition of
           $\hat\beta$ as the minimizer of $\hat R_{\hat\eta}(\beta)$, and the
           last equality is the result of adding and subtracting some
           terms. Denote (\ref{eq:T}) with $T$, (\ref{eq:R}) with $H$, and
           (\ref{eq:A}) with $Q(\tilde\alpha)$.

           Note that the assumptions of the theorem imply that
           $P_0\{|D_{\eta^\star}(O)|\leq M\}=1$ for some constant $M$. This, together
           with Lemma~\ref{lemma:bound} below, allow the application of Lemma~3
           in \cite{vanderLaan&Dudoit03} \citep[see also pages 143-145
           of][]{Dudoit&vanderLaan05} to show that
           \[\E(T+H)\lesssim \frac{1 +\log n}{n}.\]
           It remains to analyze $Q(\hat\beta)$ and $Q(\tilde\beta)$. First, we write
           $Q_\beta = Q_1(\beta) + Q_2(\beta)$, where
           \begin{align*}
             Q_1(\beta)&=(1+\delta)E_KP_0(L^0_{\hat\eta_k} -
                         L^0_{\eta^\star})(\theta_{\beta,k})\\
             Q_2(\beta)&=(1+\delta)E_K(P_{n,k}-P_0)(L^0_{\hat\eta_k} - L^0_{\eta^\star})(\theta_{\beta,k})
           \end{align*}
           For $Q_1(\beta)$, note that $R_{\eta^\star}(\beta)=E_KP_0(\theta_0  -
           \theta_{\beta,k})^2$. This yields, for $\beta\in(\hat\beta,\tilde\beta)$,
           \begin{align*}
             \E\,Q_1(\beta)&=  (1+\delta)\E\,E_KP_0(L^0_{\hat\eta_k} - L^0_{\eta^\star})(\theta_{\beta,k})\notag\\
                           &=2(1+\delta)\E\,E_KP_0(D_{\hat\eta_k}
                             - D_{\eta^\star})(\theta_0 -\theta_{\beta,k})\notag\\
                           &=2(1+\delta)\{\E\,E_KP_0(D_{\hat\eta_k}
                             - D_{\eta_0})(\theta_0 -\theta_{\beta,k})-\E\,E_KP_0(D_{\eta^\star}
                             - D_{\eta_0})(\theta_0 -\theta_{\beta,k})\}\notag
           \end{align*}
           Conditioning on $W$ first, from the definition of $\eta^\star$,
           Lemma~\ref{lemma:eif} shows that the second term in the right hand
           side is zero.  Conditioning on $W$ first along with the proof of
           Lemma~\ref{lemma:eif} and the Cauchy-Schwartz inequality also yields
           \begin{align*}
             \E\,Q_1(\beta)&=  2(1+\delta)\E\,E_KP_0(D_{\hat\eta_k}
                             - D_{\eta_0})(\theta_0 -\theta_{\beta,k})\notag\\
                           &= \sum_{t=1}^m\E\,E_KP_0(\theta_0 -\theta_{\beta,k})\left[-\frac{S(m)}{S(t)}S_0(t-1)\{h_0(t)-h(t)\}\right.\\&\left.\left\{\frac{g_{A,0}}{g_AG(t)}\sum_{k=0}^{t-1}G_0(k)\{g_{R,0}(k)
                                                                                                                                          - g_R(k)\}\frac{G(t)}{G(k+1)}
                                                                                                                                          + \frac{1}{g_A}(g_{A,0}-g_A)\right\}\,\bigg
                                                                                                                                          |\, W \right]\\
                           &\leq 2(1+\delta)\left[\E\,E_KP_0(\theta_0
                             -\theta_{\beta,k})^2\right]^{1/2}||\E(\hat g -
                             g_0)(\hat h - h_0)||\\
                           &\lesssim \sqrt{\E\tilde
                             R_{\eta^\star}(\hat \beta)}\, B_1(\hat\eta, \eta_0)
           \end{align*}
           where the last inequality follows from Lemma~\ref{lemma:bound} in
           Appendix~\ref{sec:lemmas} and the definition of $\tilde\beta$ as the
           minimizer of $\tilde R_{\eta^\star}(\beta)$, and the second to last inequality
           follows from Cauchy-Schwartz applied to the norm defined by the inner
           product $<f_k,g_k>= \E\,E_KP_0 f_k g_k$. For $Q_2(\beta)$, note that
           $(P_{n,k}-P_0)(L_{\hat\eta_k} - L_{\eta^\star})(\theta_{\beta,k})$ is an
           empirical processes with index set $\mathcal B_n$, where the latter
           set is finite. We will apply the following inequality for empirical
           processes with finite index set:
           \begin{equation}E\max_{f\in \mathcal F}|(P_n-P_0)f|\lesssim \sqrt{\frac{\log
                 |\mathcal F|}{n}}||F||,\label{eq:eresult}
           \end{equation} where $F$ is an envelope of
           $\mathcal F$. This result is a direct consequence of Lemma 19.38 of
           \cite{vanderVaart98}. Note that the all functions in
           $\mathcal F_k = \{(L^0_{\hat\eta_k} -
           L^0_{\eta^\star})(\theta_{\beta,k}):\beta \in \mathcal B_n\}$ satisfy
           \begin{align*}
             P_0(L^0_{\hat\eta_k} - L^0_{\eta^\star})^2(\theta_{\beta,k})&=P_0\{(D_{\hat\eta_k}
                                                                           - D_{\eta^\star})^2(\theta_0
                                                                           -
                                                                           \theta_{\beta,k})^2\}\\
                                                                         &\lesssim P_0(D_{\hat\eta_k}
                                                                           - D_{\eta^\star})^2\\
                                                                         &\lesssim B_2(\hat\eta,\eta_1),
           \end{align*}
           where the second inequality follows from Lemma~\ref{lemma:bound}. Thus,
           the envelope $F_k$ of $\mathcal F_k$ is bounded by the same
           quantity. This, together with (\ref{eq:eresult}) shows
           \[\E Q_2(\beta)\lesssim \sqrt{\frac{\log n}{n}}B_2(\hat\eta,\eta_1).\]
           This proves
           \begin{equation*}0\leq \E\tilde
             R_{\eta_0}(\hat \beta)\lesssim
             (1+2\delta)\E\tilde
             R_{\eta_0}(\tilde\beta) +
             \frac{1+\log n}{n} + \sqrt{\E\tilde
               R_{\eta_0}(\hat \beta)}B_1(\hat\eta,\eta_1) + \sqrt{\frac{\log
                 n}{n}}B_2(\hat\eta,\eta_1),
           \end{equation*}
           which is equivalent to $x^2-bx\leq c$ for
           \begin{align*}
             x&=\sqrt{\E\tilde R_{\eta_0}(\hat \beta)}\\
             b&=B_1(\hat\eta,\eta_1)\\
             c&=(1+2\delta)\E\tilde
                R_{\eta_0}(\tilde\beta) +
                \frac{1+\log n}{n} + \sqrt{\frac{\log
                n}{n}}B_2(\hat\eta,\eta_1).
           \end{align*}
           The quadratic formula $x\leq (b+\sqrt{b^2+4c})/2$ implies $x\leq b
           +\sqrt{c}$, which yields
           \begin{equation}0\leq \sqrt{\E\tilde
               R_{\eta_0}(\hat \beta)}\lesssim
             \sqrt{(1+2\delta)\E\tilde
               R_{\eta_0}(\tilde\beta)} + \sqrt{\frac{1+\log
                 n}{n}} + B_1(\hat\eta,\eta_0) +
             \left[\frac{\log n}{n}\right]^{1/4}\sqrt{B_2(\hat\eta,\eta_0)}
             \label{eq:orbeta}
           \end{equation}
           From our definitions and
           assumptions, the function $f(\beta)=R_{\eta_0}(\hat\theta_\beta)$
           satisfies the Lipschitz condition
           \[||f(\beta) - f(\alpha)||_{\infty}\lesssim ||\beta - \alpha||_2,\]
           where $||\cdot||_{\infty}$ denotes the supremum norm and $||\cdot||_2$ the
           Euclidean norm. Thus $f(\hat\beta)-f(\hat\alpha)$ and
           $f(\tilde\beta)-f(\tilde\alpha)$ are both bounded by $n^{-q}$, which
           allows us to replace $(\hat\beta,\tilde\beta)$ by $(\hat\alpha,\tilde\alpha)$ in (\ref{eq:orbeta}),
           completing the proof of the theorem.
         \end{proof}
         \subsubsection{Theorem \ref{theo:oraclezo}}
         \begin{proof}

           For convenience in the calculations we use the loss function
           \[L_\eta(o;f)=-D_\eta(o)d_f(z) = -\one\{D_\eta(o) > 0\}D_\eta(o) +
             |D_\eta(o)|\one[\one\{D_\eta(o) > 0\}\neq d_f],\]
           which is equivalent to the one used in the Theorem.
           Let $\eta^\star$, $L_\eta^0(O;f)$, $\hat R_\eta(\beta)$, and $\tilde R_\eta(\beta)$
           be defined as in the proof of Theorem~\ref{theo:oracle}. We have
           \begin{eqnarray*}
             0&\leq& \tilde R_{\eta^\star}(\hat \beta)\\
              &=&\tilde R_{\eta^\star}(\tilde \beta)\\
              &&+\{\tilde R_{\hat\eta}(\tilde\beta) - \tilde
                 R_{\eta^\star}(\tilde\beta)\}\\
              &&+\{\hat R_{\hat\eta}(\hat\beta) - \tilde
                 R_{\hat\eta}(\tilde\beta)\}\\
              &&+\{\tilde R_{\hat\eta}(\hat\beta) - \hat
                 R_{\hat\eta}(\hat\beta)\}\\
              &&-\{\tilde R_{\hat\eta}(\hat\beta) - \tilde
                 R_{\eta^\star}(\hat\beta)\}.
           \end{eqnarray*}
           Define
           \begin{eqnarray*}
             T(\beta)&=&(\hat R_{\hat\eta} - \tilde R_{\hat\eta})(\beta)\\
             Q(\beta)&=&(\tilde R_{\hat\eta} - \tilde R_{\eta^\star})(\beta).
           \end{eqnarray*}
           Since, by definition, $\hat R_{\hat\eta}(\hat\beta)\leq \hat
           R_{\hat\eta}(\tilde\beta)$, we have
           \[  0\leq \tilde R_{\eta^\star}(\hat \beta) + T(\tilde\beta) -
             T(\hat\beta) +Q(\tilde\beta) - Q(\hat\beta).\]
           \cite{vanderLaan&Dudoit03}, page 26, show that
           \[\E T(\tilde\beta) - \E T(\hat\beta)\lesssim (\log n/n)^{1/2}.\]
           In the proof of Theorem~\ref{theo:oraclezo}, we show that
           \[\E Q(\tilde\beta) - \E Q(\hat\beta)\lesssim B_1(\hat\eta,\eta_0),\]
           completing the proof of first claim of the theorem.

           Assume now condition \ref{ass:ma} holds with $\alpha=\infty$ such that
           $\inf_{z\in \mathbf Z}|\theta_0(z)|>0$ The proof in this case has the
           same steps as the proof of Theorem~\ref{theo:oracle} and we will only
           provide a sketch. The conditions of the Theorem allow application of
           Lemma~\ref{lemma:boundzo} below to obtain
           \[\E(T+H)\lesssim \frac{1 +\log n}{n}.\]
           For $Q_1(\alpha)$ and $Q_2(\alpha)$ we get
           \begin{align*}
             \E\,Q_1(\alpha)&=  (1+\delta)\E E_KP_0(L^0_{\hat\eta_k} -
                              L^0_{\eta^\star})(f_{\alpha,k})\\
                            &=2(1+\delta)\E\,E_KP_0(D_{\hat\eta_k}
                              - D_{\eta_0})(d_{f_{\alpha,k}} - d_0)\notag\\
                            &= \sum_{t=1}^m \E\,E_KP_0(d_{f_{\alpha,k}} - d_0)\left[-\frac{S(m)}{S(t)}S_0(t-1)\{h_0(t)-h(t)\}\right.\\&\left.\left\{\frac{g_{A,0}}{g_AG(t)}\sum_{k=0}^{t-1}G_0(k)\{g_{R,0}(k)
                                                                                                                                        - g_R(k)\}\frac{G(t)}{G(k+1)}
                                                                                                                                        + \frac{1}{g_A}(g_{A,0}-g_A)\right\}\,\bigg
                                                                                                                                        |\, W \right]\\
                            &\leq 2(1+\delta)\left[\E\,E_KP_0(\theta_0
                              -\theta_{\beta,k})^2\right]^{1/2}\E||(\hat g -
                              g_0)(\hat h - h_0)||\\
                            &\lesssim  B_1(\hat\eta, \eta_0).
           \end{align*}
           For $Q_2(\alpha)$, note that $(P_{n,k}-P_0)(L_{\hat\eta_k} -
           L_{\eta^\star})(f_{\alpha,k})$ is an empirical processes with
           index set $\mathcal A_n$, where the latter set is the finite set with
           $Mn^q$ points in which $\hat\alpha$ is computed. We will apply
           inequality (\ref{eq:eresult}). Note that the all functions in
           $\mathcal F_k = \{(L^0_{\hat\eta_k} -
           L^0_{\eta^\star})(f_{\alpha,k}):\alpha \in \mathcal A_n\}$ satisfy
           \begin{align*}
             P_0(L^0_{\hat\eta_k} - L^0_{\eta^\star})^2(\theta_{\alpha,k})&=P_0\{(D_{\hat\eta_k}
                                                                            - D_{\eta^\star})^2(d_{f_{\alpha,k}} - d_0)^2\}\\
                                                                          &\lesssim P_0(D_{\hat\eta_k}
                                                                            - D_{\eta^\star})^2\\
                                                                          &\lesssim B_2^2(\hat\eta,\eta_1),
           \end{align*}
           where the second inequality follows from Lemma~\ref{lemma:bound}. Thus,
           the envelope $F_k$ of $\mathcal F_k$ is bounded by the same
           quantity. This, together with (\ref{eq:eresult}) shows
           \[E Q_2(\alpha)\lesssim \sqrt{\frac{\log n}{n}}B_2(\hat\eta,\eta_1).\]
           This completes the proof.
         \end{proof}
         \subsubsection{Lemma~\ref{lemma:suloss}}
         \begin{proof}
           This is a direct application of Theorems 1 (part 3) and 2 of
           \cite{bartlett2006convexity}. See also Theorem 5 of
           \cite{luedtke2016super}.
         \end{proof}

         \subsubsection{Lemmas}\label{sec:lemmas}
         \begin{lemma}\label{lemma:bound}
           Consider the assumptions of Theorem~\ref{theo:oracle}. Let
           $Z=L_{\eta_0}(O;\theta) - L_{\eta_0}(O;\theta_0)$. We have
           \[\var_0(Z)\lesssim
             E_0(Z)\]
         \end{lemma}
         \begin{proof}
           First, note that
           \[Z = \{\theta_0(Z) -
             \theta(Z)\}\{2D_{\eta_0}(O)-\theta(Z)-\theta_0(Z)\}.\]
           In light of Lemma~\ref{lemma:eif} we have
           \[E_0(Z)=E_0\{\theta_0(Z) - \theta(Z)\}^2.\]
           Note that $P_0\{|2D_{\eta_0}-\theta(Z)-\theta_0(Z)|\leq 4\max(M,C_1)\}=1$. Thus
           \begin{align*}
             \var_0(Z) &\leq  E_0(Z^2)\\
                       &= E\{\theta_0(Z) -
                         \theta(Z)\}^2\{2D_{\eta_0}(O)-\theta(Z)-\theta_0(Z)\}^2\\
                       &\leq 16\max(M^2,C_1^2)E_0(Z),
           \end{align*}
           which completes the proof of the lemma.
         \end{proof}
         \begin{lemma}\label{lemma:boundzo}
           Consider the assumptions of Theorem~\ref{theo:oraclezo}. Let
           \[L_{\phi,\eta}(o,f)=D_\eta(o)d_f(z),\] Let $Z=L_{\eta_1}(O;\theta)
           - L_{\eta_1}(O;\theta_0)$. We have
           \[\var_0(Z)\lesssim
             E_0(Z)\]
         \end{lemma}
         \begin{proof}
           We have
           \begin{align*}
             E_0[Z^2]&=E_0|d_{0}(Z) - d_f(Z)|^2D^2_{\eta_1}\\
                     &\leq C E_0\one\{d_0(Z)\neq d_f(Z)\}\\
                     &\leq C
                       E_0\frac{|\theta_0(Z)|}{\inf_z|\theta_0(Z)|}\one\{d_0(Z)\neq d_f(Z)\}\\
                     &\leq C_2 E_0|\theta_0(Z)|\one\{d_0(Z)\neq d_f(Z)\}\\
                     &=C_2 E_0(Z).
           \end{align*}
         \end{proof}

         \begin{lemma}\label{lemma:bound}
           For each $\hat\eta=(\hat g,\hat h)\to \eta_1=(g_1,h_1)$ such that either
           $g_1=g_0$ or $h_1=h_0$ define
           \[\eta^\star =
             \begin{cases}
               (g_0,\hat h)&\text{if } g_1 = g_0\text{ and } h_1 \neq h_0\\
               (\hat g, h_0)&\text{if } g_1 \neq g_0\text{ and } h_1 = h_0\\
               (g_0, h_0)&\text{if } g_1 = g_0 \text{ and } h_1 = h_0.
             \end{cases}\]
           We have
           \[P_0(D_{\hat \eta} - D_{\eta^\star})^2 \lesssim B^2(\hat \eta,\eta_1),\]
           with $B^2$ defined in Theorem~\ref{theo:oracle}.
         \end{lemma}
         \begin{proof}
           First let $g_1=g_0$ and $h_1 \neq h_0$. Then $\eta^\star=(g_0,\hat
           h)$ and straightforward algebra shows
           \[P_0(D_{\hat\eta} - D_{\eta^\star})^2 \lesssim ||\hat g- g_1||^2\]
           Analogously, for $g_1\neq g_0$ and $h_1 = h_0$ we have
           \[P_0(D_{\hat\eta} - D_{\eta^\star})^2 \lesssim ||\hat h- h_1||^2.\]
           Now, for $g_1= g_0$ and $h_1 = h_0$ we get
           \[P_0(D_{\hat\eta} - D_{\eta^\star})^2 \lesssim \{||\hat h- h_1|| +
             ||\hat g- g_1||\}^2.\]
           Putting these results together proves the lemma.
         \end{proof}
         \begin{lemma}\label{lemma:telescope}
           For two sequences $a_1,\ldots,a_m$ and $b_1,\ldots,b_m$ we have
           \[\prod_{t=1}^{m}(1-a_t) - \prod_{t=1}^{m}(1-b_t) = \sum_{t=1}^{m}\left\{\prod_{k=1}^{t-1}(1-a_k)(b_t-a_t)\prod_{k=t+1}^{m}(1-b_k)\right\}.\]
         \end{lemma}
         \begin{proof}
           Replace $(b_t-a_t)$ by $(1 - a_t) - (1 - b_t)$ in the right hand
           side and expand the sum to notice it is a telescoping sum.
         \end{proof}

\end{document}